\documentclass[conference,compsoc]{IEEEtran}
\pdfoutput=1
\usepackage[utf8]{inputenc}
\usepackage[ruled, linesnumbered]{algorithm2e}
\usepackage{pgfplots}
\usepackage{amsmath}
\usepackage{xspace}
\usepackage{hyperref}%references are links
\usepackage{subcaption}
\usepackage{amsmath}
\usepackage{amsthm}

\newcommand{\bigO}{\mathcal{O}}

\title{Generating massive complex networks \\ with hyperbolic geometry faster in practice}
\author{\IEEEauthorblockN{Moritz von Looz}
\IEEEauthorblockA{% Institute of Theoretical Informatics\\
Karlsruhe Institute of Technology (KIT),
Germany\\
Email: moritz.looz-corswarem@kit.edu}
\and
\IEEEauthorblockN{Mustafa Safa Özdayi}
\IEEEauthorblockA{
Istanbul Technical University,
Turkey\\
Email: ozdayi@itu.edu.tr
}
\and
\IEEEauthorblockN{Sören Laue}
\IEEEauthorblockA{Friedrich Schiller University Jena,
Germany\\
Email: soeren.laue@uni-jena.de}
\and
\IEEEauthorblockN{Henning Meyerhenke}
\IEEEauthorblockA{% Institute of Theoretical Informatics\\
Karlsruhe Institute of Technology (KIT),
Germany\\
Email: meyerhenke@kit.edu}

}

\newcommand{\asinh}{\ensuremath{\mathrm{asinh}}}
\newtheorem{theorem}{Theorem}

\begin{document}

  \maketitle

  \begin{abstract}
    Generative network models play an important role in algorithm development, scaling studies, network analysis, and realistic system benchmarks for graph data sets.
   The commonly used graph-based benchmark model R-MAT has some drawbacks concerning realism and the scaling behavior of network properties.
    A complex network model gaining considerable popularity builds random hyperbolic graphs, generated by distributing points within a disk in the hyperbolic plane and then adding edges between points whose hyperbolic distance is below a threshold. 
    We present in this paper a fast generation algorithm for such graphs. 
    Our experiments show that our new generator achieves speedup factors of 3-60 over the best previous implementation.
    One billion edges can now be generated in under one minute on a shared-memory workstation.
    Furthermore, we present a dynamic extension to model gradual network change, while preserving at each step the point position probabilities.
  \end{abstract}

  \section{Introduction}
  \label{sec:introduction}
  Relational data of complex relationships often take the form of \emph{complex networks}, graphs with heterogeneous and often hierarchical structure, low diameter, high clustering, and a heavy-tailed degree distribution.
  Examples include social networks, the graph of hyperlinks between websites, protein interaction networks, and infrastructure routing networks on the autonomous system level~\cite{newman2010networks}.

Frequently found properties in generative models for complex network are non-negligible clustering (ratio of triangles to triads), a community structure, and a heavy-tailed degree distribution~\cite{chakrabarti2006graph}, such as a power-law.

  Benchmarks developed to evaluate a system with respect to floating point operations do not represent the requirements of graph algorithms, especially with heterogeneous datasets such as complex networks. The Graph500 benchmark~\cite{graph500} addresses this gap; it is the most widely-used graph benchmark in high-performance computing. It uses the \emph{Recursive Matrix} (R-MAT)~\cite{CZF04} model to generate synthetic networks as benchmark instances.
  Graphs from this model are efficiently computable, but suffer from drawbacks in terms of realism.
  For example, even with fixed parameters, the clustering coefficient shrinks with graph size, while the number of connected components increases, which is problematic for scaling studies~\cite{KoPiPlSe14}.

  An interesting model without this problem are \emph{random hyperbolic graphs} (RHG), a family of geometric graphs in the hyperbolic plane.
  Krioukov et al.~\cite{Krioukov2010} introduced this graph model and showed how the structure of complex networks naturally develops from the properties of hyperbolic geometry.
  To generate a RHG, one randomly samples node positions in a hyperbolic disk, then connects two nodes with an edge with a probability depending on their hyperbolic distance.
  In a special case of this model, an edge between two nodes is added exactly if their distance is below a threshold. This subset of RHG, sometimes called \emph{threshold random hyperbolic graphs}, is well-analyzed theoretically~\cite{bode2014probability,DBLP:conf/icalp/GugelmannPP12,raey} and could be considered as unit-disk graphs in hyperbolic space.
  The resulting graphs show a power-law degree distribution with adjustable exponent, provably high clustering~\cite{DBLP:conf/icalp/GugelmannPP12}, and small diameter~\cite{raey}.

\subsubsection*{Motivation, outline, and contribution}
  A fast generator implementation that scales to large graph sizes and provides
sufficient realism is necessary to create meaningful graph benchmark instances in acceptable time.
While our previous work~\cite{Looz2015HRG} was able to improve the quadratic time complexity of the pairwise probing approach~\cite{Aldecoa2015} for threshold RHGs, it still has superlinear time complexity.
  We therefore provide a faster generation algorithm in this paper for threshold random hyperbolic graphs (Section~\ref{sec:algorithm}), using a new spatial data structure.
  The key idea is to divide the relevant part of the hyperbolic plane into ring-shaped slabs and use these to bound the coordinates of possible neighbors in each slab.
  As our experiments (Section~\ref{sec:experiments}) show, a network with 10 million vertices and 1G edges can be generated with our shared-memory parallel implementation in under one minute, yielding a speedup factor of up to 60 over the best previous implementation~\cite{Looz2015HRG}. For a graph with $n$ nodes and $m$ edges, the measurements suggest 
an $\bigO(n \log n + m)$ time complexity, but we do not have a proof for this.

  While an algorithm with optimal expected linear time complexity has been suggested in a theoretical 
paper~\cite{bringmann2015geometric}, our present work provides the fastest implementation to date.
  The generator code is publicly available in our network analysis toolkit NetworKit~\cite{staudt2014networkit}.
  
  \section{Related Work}
  \label{sec:related-work}
\subsubsection*{Generative Models}
Due to the growing interest in complex networks, numerous generators for them exist. For a comprehensive overview, which
would be outside the scope of this paper, we refer the interested reader to Goldenberg's survey~\cite{goldenberg2010survey}. 
None of the models is suitable for all use cases. As mentioned above, the \emph{Recursive Matrix} (R-MAT)~\cite{CZF04} model 
has received particular attention in the HPC community due to its use in the Graph500 benchmark~\cite{graph500}.
%The model recreates a power-law degree distribution by recursively subdividing the adjacency matrix and adding edges to the matrix sections with parametrized probabilities.
%  Graphs with $n$ vertices and $m$ edges can be generated in $\Theta(m\log n)$.

  \subsubsection*{Hyperbolic Geometry}
  Hyperbolic space is one of the three isotropic spaces, the other two being the (more common) Euclidean space and spherical space. In contrast to the flat Euclidean geometry and the positively curved spherical geometry, hyperbolic geometry has negative curvature~\cite{Anderson2005}.
  Among other interesting properties, hyperbolic geometry shows an \emph{exponential expansion of space}:
  While the area of a Euclidean circle grows quadratically with the circle radius, the area of a circle on the hyperbolic plane grows exponentially with its radius.
  In balanced trees, the number of nodes at a certain distance from the root also grows exponentially with said distance, leading to the suggestion that hierarchical complex networks with tree-like structures might be easily embeddable in hyperbolic space~\cite{Krioukov2010}.
  Indeed, Bogu{\~n}{\'a} et al.~\cite{Boguna2010} demonstrate the connection between hyperbolic geometry and complex networks by embedding the autonomous system internet graph in the hyperbolic plane and enabling locally greedy routing.

  As a generative model, Krioukov et al.~\cite{Krioukov2010} introduced random hyperbolic graphs in 2010.
  To generate a graph, points are first distributed randomly within in a disk $\mathcal{D}_R$ of radius $R$ in the hyperbolic plane.
  The probability density functions for the point distributions are given in polar coordinates, the angular coordinate $\phi$ is distributed uniformly over $[0,2\pi]$, the radial coordinate $r$ is given by~\cite[Eq.~(17)]{Krioukov2010}:
  \begin{equation}
  f(r) = \alpha\frac{\sinh(\alpha r)}{\cosh(\alpha R)-1}
  \label{eq:base-radial-distribution}
  \end{equation}
  The parameter $\alpha$ governs node dispersion, which determines the power-law exponent of the resulting degree distribution.

  After sampling point positions, edges are then added to each node pair $(u,v)$ with a probability given in~\cite[Eq.~(41)]{Krioukov2010}, depending on their hyperbolic distance and parametrized by a temperature $T\geq 0$:
  \begin{equation}
  f(x) = \frac{1}{e^{(1/T)\cdot(x-R)/2}+1}
  \label{eq:Krioukov-equation}
  \end{equation}

  For $\alpha \geq \frac{1}{2}$, the resulting degree distribution follows a power law with exponent $\gamma := 2\alpha +1$~\cite[Eq.~(29)]{Krioukov2010}.
  Given two points in polar coordinates $p=(\phi_p, r_p), q=(\phi_q, r_q)$ on the hyperbolic plane, the distance between them is given by the hyperbolic law of cosines:
  \begin{equation}
  \cosh{\mathrm{dist}(p,q)} = \cosh{r_p}\cosh{r_q}-\sinh{r_p}\sinh{r_q}\cos{|\phi_p-\phi_q|}
  \label{eq:hyperbolic-cosines}
  \end{equation}

  As mentioned briefly in Section~\ref{sec:introduction}, an important special case is $T=0$, where an edge is added to a node pair exactly if the hyperbolic distance between the points is below a threshold.
  This graph family is sometimes called \emph{threshold random hyperbolic graphs}, \emph{hyperbolic unit-disk graphs} or (slightly confusingly) just random hyperbolic graphs.
  While we consider \emph{hyperbolic unit-disk graphs} to be more precise, we stick with \emph{threshold random hyperbolic graphs} to avoid name proliferation.
  Many theoretical results are for this special case~\cite{raey}.

\subsubsection*{RHG Generation Algorithms}
  Previous generators for random hyperbolic graphs exist, both for the general and special case. Aldecoa et al.~\cite{Aldecoa2015} present a generator for the general case with quadratic time complexity,
  calculating distances and sampling edges for all $\Theta(n^2)$ node pairs.

  Von Looz et al.~\cite{Looz2015HRG} use polar quadtrees to generate threshold RHGs with a time complexity of $\bigO((n^{3/2}+m)\log n)$ (with high probability). Recently, von Looz and Meyerhenke~\cite{2015arXiv150901990V} have extended this approach to generate general RHGs with the same time complexity.

  Bringmann et al.~\cite{bringmann2015geometric} propose \emph{Geometric Inhomogeneous Random Graphs} as a generalization of RHGs and describe a generation algorithm with expected linear time complexity. To our knowledge no implementation
of this algorithm is available.

  \section{Algorithm}
  \label{sec:algorithm}
  Our main idea is to partition the hyperbolic plane into concentric ring-shaped slabs (Section~\ref{subsec:slabs}) and use them to limit the number of necessary distance calculations during edge creation (Algorithm~\ref{algo:generation}).
  Point positions are sampled, sorted by their angular coordinates and stored in the appropriate slab as determined by their radial coordinates.
  To gather the neighborhood of a point $v$, we then iterate over all slabs and examine possible neighbors within them.
  Since each slab limits the radial coordinates of points it contains, we can use Eq.~(\ref{eq:hyperbolic-cosines}) to also bound the angular coordinates of possible neighbors in each slab, thus reducing the number of comparisons and running time.

  \subsection{Data Structure}
  \label{subsec:slabs}
  Let $C$ = $\{ c_0, c_1, ... c_{\text{max}} \}$ be a set of $\log n$ ordered radial boundaries, with $c_0 = 0$ and $c_{\text{max}} = R$.
  We then define a \emph{slab} $S_i$ as the area enclosed by $c_i$ and $c_{i+1}$.
  A point $p=(\phi_p, r_p)$ is contained in slab $S_i$ exactly if $c_i \leq r_p < c_{i+1}$.
  Since slabs are ring-shaped, they partition the hyperbolic disk $\mathcal{D}_R$:
  \[
  \mathcal{D}_R = \bigcup_{i=0}^{\log n} S_i.
  \]

  The choice of radial boundaries is an important tuning parameter. After experimenting with different divisions, we settled on a geometric sequence with ratio $p=0.9$.
  The relationship between successive boundary values is then: $c_{i+1} - c_{i} = p\cdot(c_i-c_{i-1})$.
  From $c_0 = 0$ and $c_{\text{max}} = R$, we derive the value of $c_1$:
  \begin{equation}
  \sum_{k=0}^{\log{}n-1} c_1 p^k = R \Leftrightarrow
  c_1\frac{1 - p^{\log{}n}}{1-p} = R \Leftrightarrow
  c_1 = \frac{(1-p)R}{1-p^{\log{}n}}
  \end{equation}
The remaining values follow geometrically.

  \begin{figure}
  \centering
  \includegraphics[width=0.8\linewidth]{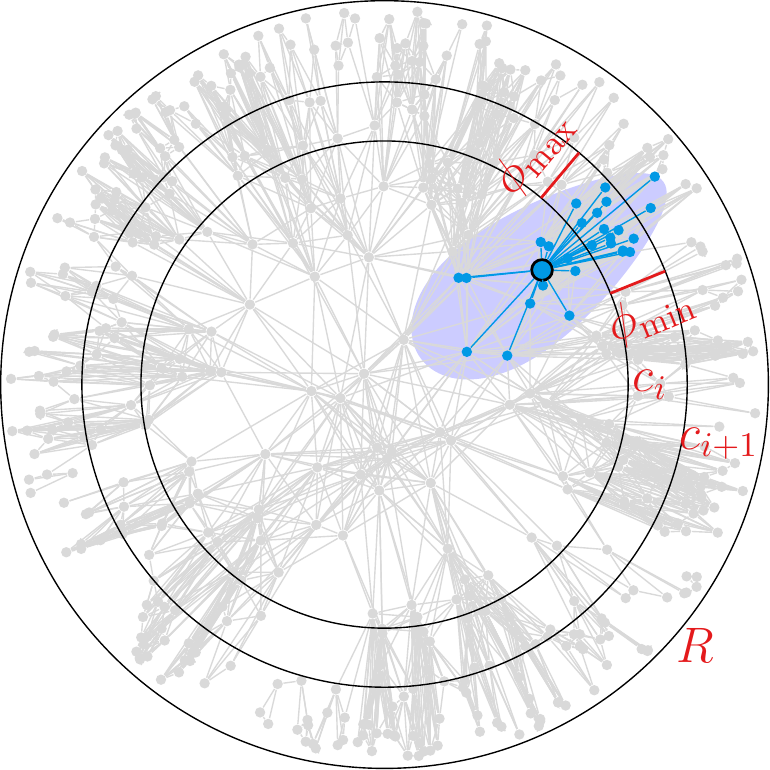}
  \caption{Graph in hyperbolic geometry with unit-disk neighborhood. Neighbors of the bold blue vertex are in the hyperbolic circle, marked in blue.}
  \label{fig:hyperbolic-graph-native} 
  \end{figure}
  Figure~\ref{fig:hyperbolic-graph-native} shows an example of a graph in the hyperbolic plane, together with slab $S_i$.
  The neighbors of the bold blue vertex $v$ are those within a hyperbolic circle of radius $R$ (0.2$R$ in this visualization), marked by the blue egg-shaped area.
  When considering nodes in $S_i$ as possible neighbors of $v$, the algorithm only needs to examine nodes whose angular coordinate is between $\phi_{\min}$ and $\phi_{\max}$.
  
  \subsection{Generation Algorithm}

  \begin{algorithm}
    \KwIn{number of vertices $n$, average degree $\overline{k}$, power-law exponent $\gamma$}
    \KwOut{$G=(V,E)$}
    $\alpha$ = $(\gamma-1)/2$\;
    $R$ = getTargetRadius($n, \overline{k}, \alpha$)\;\label{line:targetRadius}
    $V$ = $n$ vertices\;
    $C$ = $\{ c_0, c_1, ... c_{\text{max}} \}$ set of $\log n$ ordered radial coordinates, with $c_0 = 0$ and $c_{\text{max}} = R$\;\label{line:generate-radials}
    $B$ = $\{ b_0, b_1, ... b_{\text{max}} \}$  set of $\log n$ empty sets\;\label{line:define-bands}

    \ForPar{vertex $v\in V$}{\label{line:generate-vertices-begin}
    draw $\phi[v]$ from $\mathcal{U}[0,2\pi)$\;\label{line:draw-angular}
    draw $r[v]$ with density $f(r) = \alpha\sinh(\alpha r)/(\cosh(\alpha R)-1)$\;\label{line:draw-radial}
    insert $(\phi[v], r[v])$ in suitable $b_i$ so that $c_i \leq r[v] \leq c_{i+1}$\;\label{line:insert-band}
    }\label{line:generate-vertices-end}

    \ForPar{$b \in B$}{\label{line:sort-band-begin}
    sort points in $b$ by their angular coordinates\;
    }\label{line:sort-band-end}

    \ForPar{vertex $v\in V$}{\label{line:last-for-begin}
    \For{band $b_i \in B$, where $c_{i+1} > r[v]$}{ \label{line:inner-last-for-begin}
    $\min_\phi$, $\max_\phi$ = getMinMaxPhi($\phi[v], r[v]), c_i, c_{i+1}, R$)\;\label{line:band-boundaries}
    \For{vertex $w \in b_i$, where $\min_\phi \leq \phi[w] \leq \max_\phi$}{ \label{line:slab-vertices}
    \If{$\text{dist}_{\mathcal{H}}(v,w) \leq R$}{ \label{line:edge-loop}
    add $(v,w)$ to $E$\; \label{line:edge-insertion}
    }
    }
    }
    }\label{line:last-for-end}
    \Return{$G$}\;
    \caption{Graph Generation}
    \label{algo:generation}
    \end{algorithm}

    \subsubsection*{Algorithm}
    Algorithm~\ref{algo:generation} shows the generation of $G = (V, E)$ with average degree $k$ and power-law exponent $\gamma$.
    First, the radius $R$ of the hyperbolic disk is calculated according to desired graph size and density (Line~\ref{line:targetRadius}).

    \subsubsection*{getTargetRadius}
    This function is unchanged from our previous work~\cite{Looz2015HRG}.
    For given values of $n, \alpha$ and $R$, an approximation of the expected average degree $\overline{k}$ is given by~\cite[Eq.~(22)]{Krioukov2010} and the notation $\xi=(\alpha/\zeta)/(\alpha/\zeta-1/2)$:
      \begin{align}
     &\overline{k} = \frac{2}{\pi} \xi^2 n\cdot e^{-\zeta R/2}  + \frac{2}{\pi} \xi^2 n \\
     &\cdot\left(e^{-\alpha R} \left(\alpha \frac{R}{2} \left(\frac{\pi}{4} \left(\frac{\zeta}{\alpha}^2\right) -(\pi-1) \frac{\zeta}{\alpha} + (\pi-2) \right) -1 \right) \right)
    \end{align}
    The value of $\zeta$ can be fixed while retaining all degrees of freedom in the model~\cite{bode2014probability}, we thus assume $\zeta=1$.
    We then use binary search with fixed $n, \alpha$ and desired $\overline{k}$ to find an $R$ that gives us a close approximation of the desired average degree $\overline{k}$.
    Note that the above equation is only an approximation and might give wrong results for extreme values.
    Our implementation could easily be adapted to skip this step and accept the commonly used~\cite{kiwi2015bound} parameter $C$, with $R = 2\ln n + C$ or even accept $R$ directly.
    For increased usability, we accept the average degree $\overline{k}$ as a parameter in the default version.

    \subsubsection*{Vertex Positions and Bands}
    After settling the disk boundary,
    the radial boundaries $c_i$ are calculated (Line~\ref{line:generate-radials}) as defined above, the disk $\mathcal{D}_R$ is thus partitioned into $\log n$ slabs.
    For each slab $S_i$, a set $b_i$ stores the vertices located in the area of $S_i$.
    These sets $b_i$ are initially empty (Line \ref{line:define-bands}).

    The vertex positions are then sampled randomly in polar coordinates (Lines~\ref{line:draw-angular} and \ref{line:draw-radial}) and stored in the corresponding set,
    i.e, vertex $v$ is put into set $b_i$ iff $c_i \leq r[v] < c_{i+1}$ (Line~\ref{line:insert-band}).
    Within each set, vertices are sorted with respect to their angular coordinates (Lines~\ref{line:sort-band-begin} to \ref{line:sort-band-end}).
    
    \subsubsection*{getMinMaxPhi}
    The neighbors of a given vertex $v=(\phi_v, r_v)$ are those whose hyperbolic distance to $v$ is at most $R$.
    Let $b_i$ be the slab between $c_i$ and $c_{i+1}$,  and $u=(\phi_u, r_u) \in b_i$ a neighbor of $v$ in $b_i$.
    Since $u$ is in $b_i$, $r_u$ is between $c_i$ and $c_{i+1}$.
    With the hyperbolic law of cosines, we can conclude:
    \begin{align}
      \cosh{R} \geq \cosh{r_v}\cosh{c_i}-\sinh{r_v}\sinh{c_i}\cos{|\phi_u-\phi_v|} \Leftrightarrow \\
      \cosh{r_v}\cosh{c_i} - \cosh{R} \leq \sinh{r_v}\sinh{c_i}\cos{|\phi_u-\phi_v|} \Leftrightarrow\\
      \cos{|\phi_u-\phi_v|} \geq \frac{\cosh{r_v}\cosh{c_i} - \cosh{R}}{\sinh{r_v}\sinh{c_i}} \Leftrightarrow\\
      |\phi_u - \phi_v| \leq \cos^{-1}\left(\frac{\cosh{r_v}\cosh{c_i} - \cosh{R}}{\sinh{r_v}\sinh{c_i}} \right)
    \end{align}

    To gather the neighborhood of a vertex $v=(\phi_v, r_v)$, we iterate over all slabs $S_i$ and compute for each slab how far the angular coordinate $\phi_q$ of a possible neighbor in $b_i$ can deviate from $\phi_v$ (Line~\ref{line:band-boundaries}).
    We call the vertices in $b_i$ whose angular coordinates are within these bounds the \emph{neighbor candidates} for $v$ in $b_i$.
    
    Since points are sorted according to their angular coordinates, we can quickly find the leftmost and rightmost neighbor candidate in each slab using binary search.
    We then only need to check each neighbor candidate (Line~\ref{line:slab-vertices}), compute its hyperbolic distance to $v$ and add an edge if this distance is below $R$ (Lines~\ref{line:edge-loop} and \ref{line:edge-insertion}).
    Since edges can be found from both ends, we only need to iterate over slabs in one direction; we choose outward in our implementation (Line~\ref{line:inner-last-for-begin}).
    The process is repeated for every vertex $v$ (Line~\ref{line:last-for-begin}).
   
   Not surprisingly the running time of Algorithm~\ref{algo:generation} is dominated by the range queries~(Lines~\ref{line:last-for-begin}-\ref{line:last-for-end}). Our experiments in Section~\ref{sec:experiments} suggest a running time of 
   $\mathcal{O}(n \log n + m)$ for the complete algorithm. This should be seen as an empirical observation; 
   we leave a mathematical proof for future work.
    
    \subsection{Dynamic Model}
    \label{subsec:dynamic-model}
    To model gradual change in networks, we design and implement a dynamic version with node movement.
    While deleting nodes or inserting them at random positions is a suitable dynamic behavior for modeling internet infrastructure with sudden site failures or additions,
    change in e.\ g.,\ social networks happens more gradually.
   
    A suitable node movement model needs to be \emph{consistent}: After moving a node, the network may change, but properties should stay the same \emph{in expectation}.
    Since the properties emerge from the node positions, the probability distribution of node positions needs to be preserved.
    In our implementation, movement happens in discrete time steps.
    We choose the movement to be \emph{directed}: If a node $i$ moves in a certain direction at time $t$, it will move in the same direction at $t+1$,
    except if the new position would be outside the hyperbolic disk $\mathcal{D}_R$.
    In this case, the movement is inverted and the node ``bounces'' off the boundary.
    The different probability densities in the center of the disk and the outer regions can be translated into movement speed: A node is less likely to be in the center; thus it needs to spend less time there while traversing it, resulting in a higher speed.
    We implement this movement in two phases:
    In the initialization, step values $\tau_\phi$ and $\tau_r$ are assigned to each node according to the desired movement.
    Each movement step of a node then consists of a rotation and a radial movement.
    The rotation step is a straightforward addition of angular coordinates: $\mathrm{rotated}(\phi, r, \tau_\phi) = (\phi+\tau_\phi)$ mod $2\pi$.
    The radial movement is described in Algorithm~\ref{algo:radial-movement} and a visualization is shown in Figure~\ref{fig:visualization-dynamic}.
    \begin{algorithm}
    \KwIn{$r,\tau_r, R, \alpha.$}
    \KwOut{$r_{\text{new}}$}
    \begin{enumerate}
    \item x = $\sinh(r\cdot\alpha)$\;
    \item y = x+$\tau_r$\;
    \item z = $\asinh(y)/\alpha$\;
    \item \textbf{Return} z
    \end{enumerate}
    \caption{Radial movement in dynamic model}
    \label{algo:radial-movement}
    \end{algorithm}

    If the new node position would be outside the boundary ($r > R$) or below the origin ($r < 0$), the movement is reflected and $\tau_r$ set to $-\tau_r$.
    \begin{theorem}
    Let $f_{r,\phi}((p_r, p_\phi))$ be the probability density of point positions, given in polar coordinates.
    Let $\mathrm{move}((p_r, p_\phi))$ be a movement step.
    Then, the node movement preserves the distribution of angular and radial distributions:
    $f_{r,\phi}(\mathrm{move}((p_r, p_\phi))) = f_{r,\phi}((p_r, p_\phi))$.
    %$F_X(r) = F_X(\mathrm{scale}(r))$ for $0 \leq r \le R$ and $F_\Phi(\mathrm{rotated}((\phi, r))) = F_\Phi(\phi)$ for $0 \leq \phi \le 2\pi$.
    \end{theorem}
    \begin{proof}
    Since the distributions of angular and radial coordinates are independent, we consider them separately: $f_{r,\phi} (p_r, p_\phi) = f_r (p_r) \cdot f_\phi (p_\phi)$.
    
    As introduced in Eq.~(\ref{eq:base-radial-distribution}), the radial coordinate $r$ is sampled from a distribution with density $\alpha\sinh(\alpha r) / (\cosh(\alpha R) -1)$.
    %and $\tau_r$ is the unscaled radial movement parameter. %\ref{eq:radial-distribution}
    We introduce random variables $X, Y, Z$ for each step in Algorithm~\ref{algo:radial-movement}, each is denoted with the upper case letter of its equivalent. 
    An additional random variable $Q$ denotes the pre-movement radial coordinate.
    The other variables are defined as $X = \sinh(Q\cdot\alpha)$, $Y = X+\tau_r$ and $Z=\asinh(Y)/\alpha$. % = \mathrm{scale}(Q,\tau_r)$.
    
    Let $f_Q, f_X, f_Y$ and $f_Z$ denote the density functions of these variables:
    \begin{align}
     f_Q(r) &= \frac{\alpha\sinh(\alpha r)}{\cosh(\alpha R) -1}\\
     f_X(r) &= f_Q\left(\frac{\asinh(r)}{\alpha}\right) = \frac{\alpha r}{\cosh(\alpha R) -1}\\
     f_Y(r) &= f_X(r-\tau_r) = \frac{\alpha r-\tau_r}{\cosh(\alpha R) -1}\\
     f_Z(r) &= f_Y(\sinh(r\cdot\alpha)) = \frac{\alpha \sinh(\alpha r) -\tau_r}{\cosh(\alpha R) -1}\\
     &= f_Q(r) - \frac{\tau_r}{\cosh(\alpha R) -1}
    \end{align}

    The distributions of $Q$ and $Z$ only differ in the constant addition of $\tau_r / (\cosh(\alpha R)-1)$.
    Every $(\cosh(\alpha R)-1) / \tau_r$ steps, the radial movement reaches a limit (0 or $R$) and is reflected, causing $\tau_r$ to be multiplied with -1.
    On average, $\tau_r$ is thus zero and $F_Q(r)$ = $F_Z(r)$.

    A similar argument works for the rotational step: 
    While the rotational direction is unchanged, the change in coordinates is balanced by the addition or subtraction of $2\pi$ whenever the interval $[0,2\pi)$ is left, leading to an average of zero in terms of change.
    \end{proof}
    \begin{figure}
    \centering
    \includegraphics[width=0.65\columnwidth]{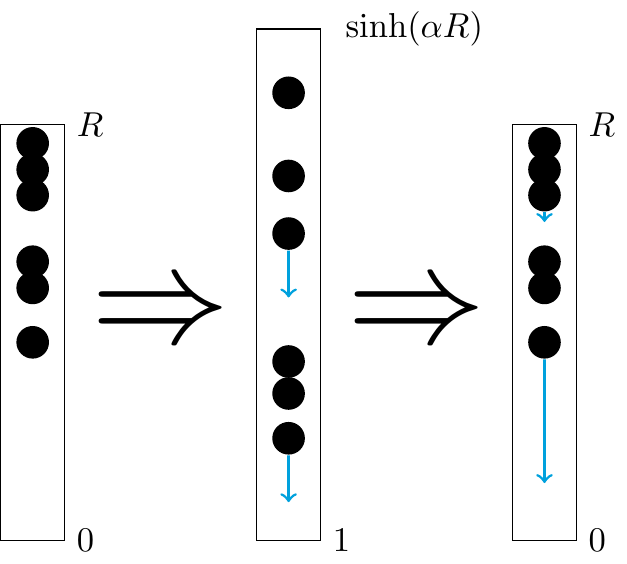}
    \caption{For each movement step, radial coordinates are mapped into the interval $[1,\sinh(\alpha R))$, where the coordinate distribution is uniform.
    Adding $\tau_r$ and transforming the coordinates back results in correctly scaled movements.}
    \label{fig:visualization-dynamic}
    \end{figure}

    \section{Experimental Evaluation}
    \label{sec:experiments}

    \subsubsection*{Setup}
    The generation algorithm is implemented in C++11 and parallelized with OpenMP. Running time measurements were made on a server with 256 GB RAM and 2x8 Intel Xeon E5-2680 cores at 2.7 GHz. With hyperthreading enabled, we use up to 32 threads.
    For memory allocations, we use the lock-free malloc implementation of Intel's Threading Building Blocks library.
    Our code is included in the network analysis toolkit NetworKit~\cite{staudt2014networkit}.

    To compare performance, we generate graphs with $10^5$, $10^6$ and $10^7$ nodes and average degrees between 1 and 64, both with the algorithm presented in this work and the implementation of von Looz et al.~\cite{Looz2015HRG}.

    To validate the distribution of generated graphs, we compare our implementation with the implementation of Aldecoa et al.~\cite{Aldecoa2015}.
    We generate graphs with $10^4$ nodes each for a combination of parameters and calculate several network analytic characteristics, averaging over 100 runs.  
    For the dynamic model, we measure the time required for a movement step and again compare the distributions of network analytic properties.

    \subsubsection*{Running Time}
    \label{subsec:performance}
    Figure~\ref{plot:time-scatter} shows the running times to generate graphs with $10^5$ to $10^7$ nodes and $2\cdot 10^5$ to $128\cdot 10^7$ edges.
    The speedup over the previously fastest implementation~\cite{Looz2015HRG} increases with graph size and sparsity, reaching up to 60 for graphs with $10^7$ nodes and $\approx 4\cdot 10^7$ edges.
    Very roughly, the experimental running times fit a complexity of $\bigO (n\log n + m)$.
    While the running times of the faster generator appear to grow more steeply with increasing edge count, this is an artifact of the logarithmic plot: The same constant increase is relatively larger compared to a smaller running time, and thus appears larger in the logarithmic drawing.

     \definecolor{markedcolor}{RGB}{31,120,180}
    \definecolor{plottinggreen}{RGB}{178,223,138}
    \definecolor{thirdhue}{RGB}{228,26,28}
    \newcommand{\aconst}{7.07}
    \newcommand{\bconst}{2.23}
    \newcommand{\cconst}{8.91}
    \newcommand{\cconstscaled}{891}
   
   \begin{figure}[h!]
      \centering
      \begin{tikzpicture}[scale=0.98]
        \begin{axis}[xmode=log,ymode=log,xlabel=edges,ylabel=running time in seconds,legend entries={}, legend style={at={(0.8,1.8)}}]
          \addplot[scatter,only marks,
          point meta = explicit symbolic,
          scatter/classes={
          %e={draw=plottinggreen,fill=plottinggreen },%
          e={draw=thirdhue,fill=thirdhue },%
          f={draw=markedcolor,fill=markedcolor},
          g={draw=orange,fill=orange},
          % 	    a={mark=diamond*, draw=plottinggreen, fill=plottinggreen},
          b={mark=diamond*, draw, thirdhue, fill=thirdhue},
          c={mark=diamond*, draw, markedcolor, fill=markedcolor},
          d={mark=diamond*, draw, orange, fill=orange}}
          ]
          table[meta=label]
          {plots/comparison-phipute-gbuild-no-leftsup-no-tbb.dat};
          %   \addlegendentry{$n=10^4$, our impl.};
          \addlegendentry{$n=10^5$, our impl.};
          \addlegendentry{$n=10^6$, our impl.};
          \addlegendentry{$n=10^7$, our impl.};
          \addlegendentry{$n=10^5$, impl. of \cite{Looz2015HRG}};
          \addlegendentry{$n=10^6$, impl. of \cite{Looz2015HRG}};
          \addlegendentry{$n=10^7$, impl. of \cite{Looz2015HRG}};
          %a*n*math.log10(n)**2 + b*m +c
          \addplot[thirdhue] expression[domain=150000:16000000] {\aconst * 10^(-3)*5 + \bconst *10^(-8)*x + \cconst*10^(-6)};\addlegendentry{$n=10^5$, theoretical fit};
          \addplot[markedcolor] expression[domain=1500000:160000000] {\aconst * 10^(-2)*6 + \bconst *10^(-8)*x + \cconst*10^(-6)};\addlegendentry{$n=10^6$, theoretical fit};
          \addplot[orange] expression[domain=15000000:1600000000] {\aconst * 10^(-1)*7 + \bconst *10^(-8)*x + \cconst*10^(-6)};\addlegendentry{$n=10^7$, theoretical fit};
        \end{axis}
      \end{tikzpicture}
      \caption{Comparison of running times to generate networks with $10^4$-$10^7$ vertices, $\alpha=1$ and varying $\overline{k}$. Circles represent running times of our implementation, diamonds the running times of the implementation of \cite{Looz2015HRG}.
      Our running times are fitted with the equation $T(n,m) = \left(\aconst \cdot n\log_{10} n + \bconst\cdot m + \cconstscaled\right)\cdot10^{-8}$ seconds.
      }
      \label{plot:time-scatter}
    \end{figure}
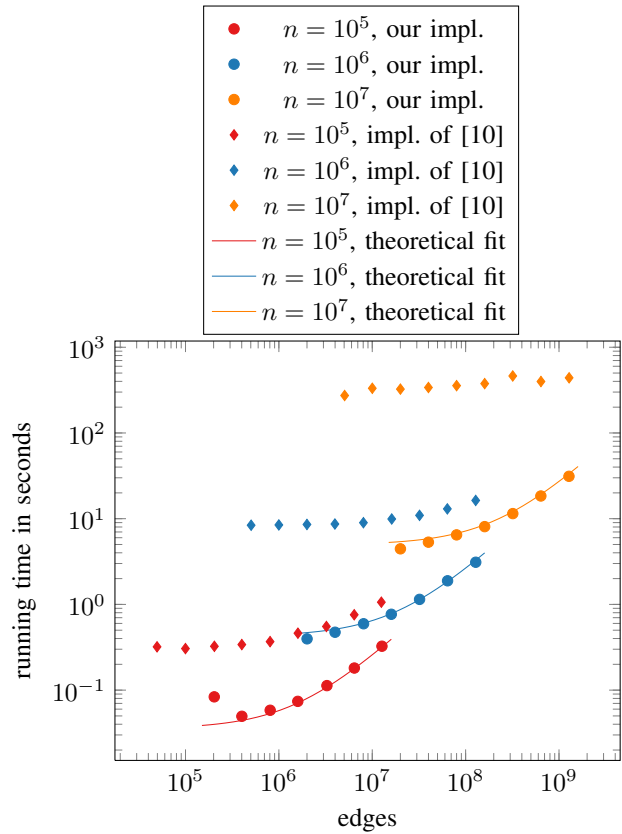

    \begin{figure}
     \begin{tikzpicture}
      \begin{axis}[xlabel=threads, ylabel=speedup factor, legend pos = north west, xtick={4,8,12,16,20,24,28,32}, extra y ticks={18.38}]
      \draw (axis cs:16,0) -- (axis cs:16,19.4);
      \addplot[markedcolor,mark=*] table [x expr = {\thisrowno{1}}, y expr = {50.94/\thisrowno{2}}] {plots/phipute-scaling.dat};\addlegendentry{total};
      \addplot[plottinggreen,mark=*] table [x expr = {\thisrowno{1}}, y expr = {47.20/\thisrowno{4}}] {plots/phipute-scaling.dat};\addlegendentry{edge sampling};
    \end{axis}
    \end{tikzpicture}
    \caption{Speedup curves for $n=10^7, k=6, \gamma=3$ on a machine with 16 physical cores (marked with a vertical line) and hyperthreading. Averaged over 10 runs.}
    \label{plot:scaling}
    \end{figure}
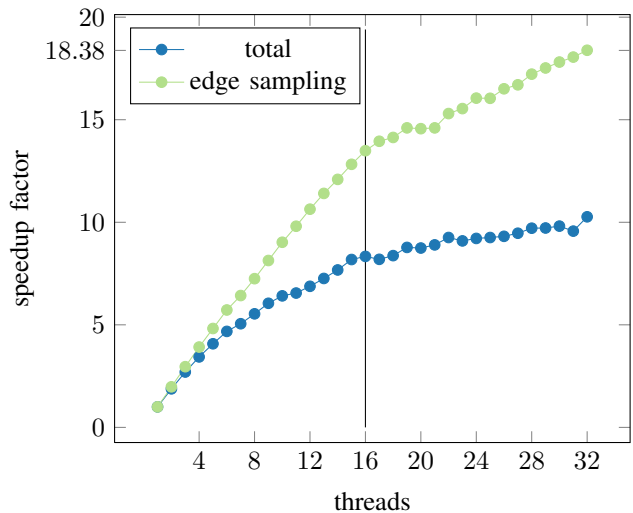

    The scaling behavior for 1 to 32 threads on 16 cores is shown in Figure~\ref{plot:scaling}.
    Considering edge sampling alone, it shows strong scaling up to the number of physical cores, with a speedup of 13.48 for 16 threads.
    With hyperthreading, the speedup increases to 18.38.
    Combining the edge lists later on into the NetworKit graph data structure, however, requires coordination and proves to be a bottleneck in parallel. If only edge lists are required, this final step can be omitted -- as done for example in the Graph500 benchmark.

    \subsubsection*{Distribution of Generated Graphs}
    The average degree assortativity, degeneracy, clustering coefficient and size and diameter of largest components of our generator and the one by Aldecoa et al.~\cite{Aldecoa2015} are shown in Plots~\ref{plot:properties-comparison-I} and \ref{plot:properties-comparison-II} in Appendix~\ref{sec:previous-impl}.
    Averaged over 100 runs, the network analytic properties show a very close match between the distributions of the two generation algorithms.
   
   \subsubsection*{Dynamic Model}
   Our implementation allows updating a graph without rebuilding it from scratch.
   Moving up to 12\% of nodes and updating an existing graph is still faster than a new static generation.
   The distribution of generated graphs is indistinguishable from the static model (Appendix~\ref{sec:dynamic-impl}).

    \section{Conclusions}
    \label{sec:conclusion}
    We have provided the fastest implementation so far to generate massive complex networks based on threshold random hyperbolic graphs. The running time improvement is particularly large for graphs with realistic densities.
%    The distribution of the generated graphs is indistinguishable from the distribution of graphs given by the original generator of Aldecoa et al.~\cite{Aldecoa2015}.
    
    We have also presented a model extension to cover gradual node movement and have proved its consistency regarding the probability densities of vertex positions. 
   
   Both the static and the dynamic model can serve as complex network generators with reasonable realism and fast generation
   times even for massive networks.

      \begin{small}
   \subsubsection*{Acknowledgements}
   This work is partially supported by German Research Foundation (DFG)
grant ME 3619/3-1 (FINCA) and grant GI-711/5-1, both within the Priority Programme 1736 \emph{Algorithms
for Big Data}.
\end{small}

    \bibliographystyle{IEEEtran}
    \bibliography{PolylogGenerator}

    \clearpage
    \appendices

    \section{Comparison with Previous Implementation~\cite{Aldecoa2015}}
    \label{sec:previous-impl}

    \begin{figure}[h!]
      \begin{subfigure}[t]{.45\linewidth}
        \includegraphics[width=\linewidth]{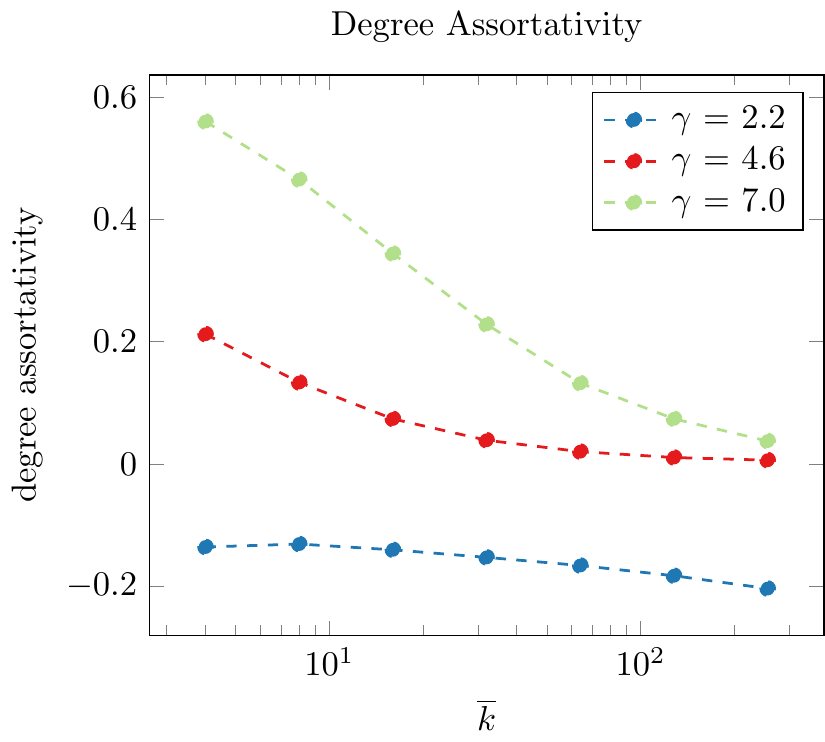}
        % \caption{Degree assortativity}
        \label{plot:native-degass}
      \end{subfigure}
      \quad
      \begin{subfigure}[t]{.45\linewidth}
        \includegraphics[width=\linewidth]{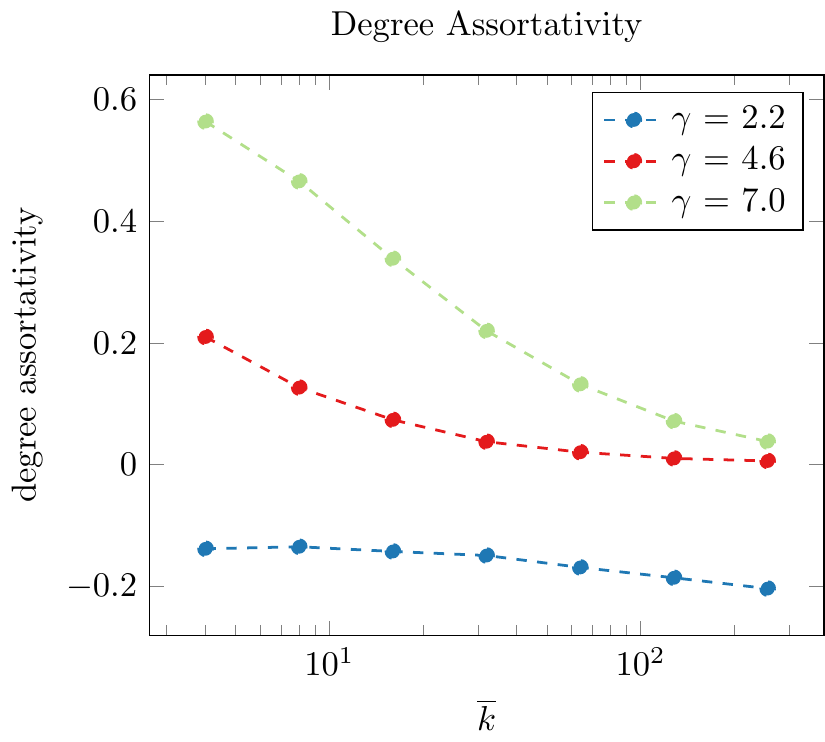}
        % \caption{Degree assortativity}
        \label{plot:comparison-degass}
      \end{subfigure}

      \begin{subfigure}[t]{.45\linewidth}
        \includegraphics[width=\linewidth]{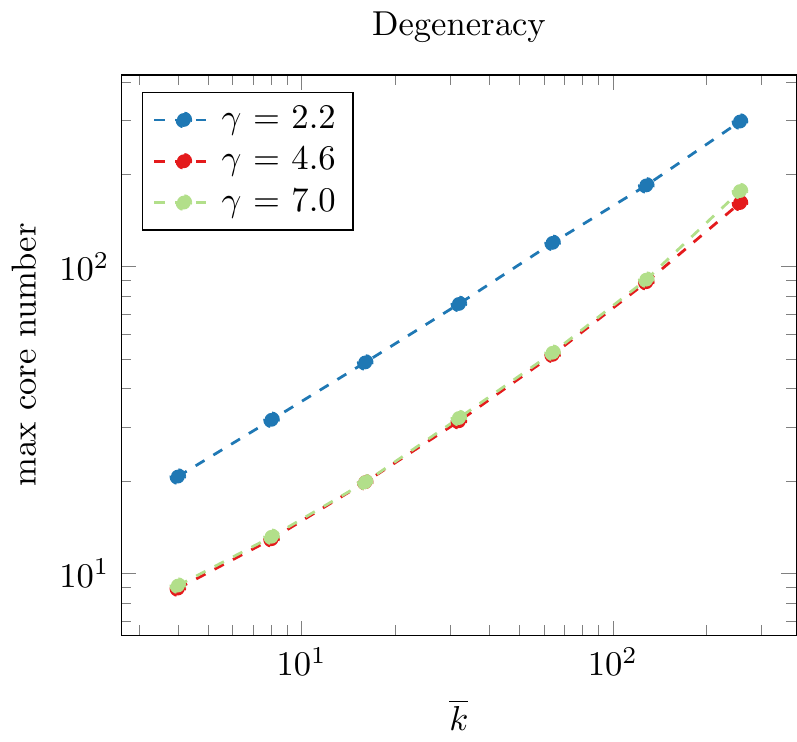}
        % \caption{Max core number}
        \label{plot:native-degen}
      \end{subfigure}
      \quad
      \begin{subfigure}[t]{.45\linewidth}
        \includegraphics[width=\linewidth]{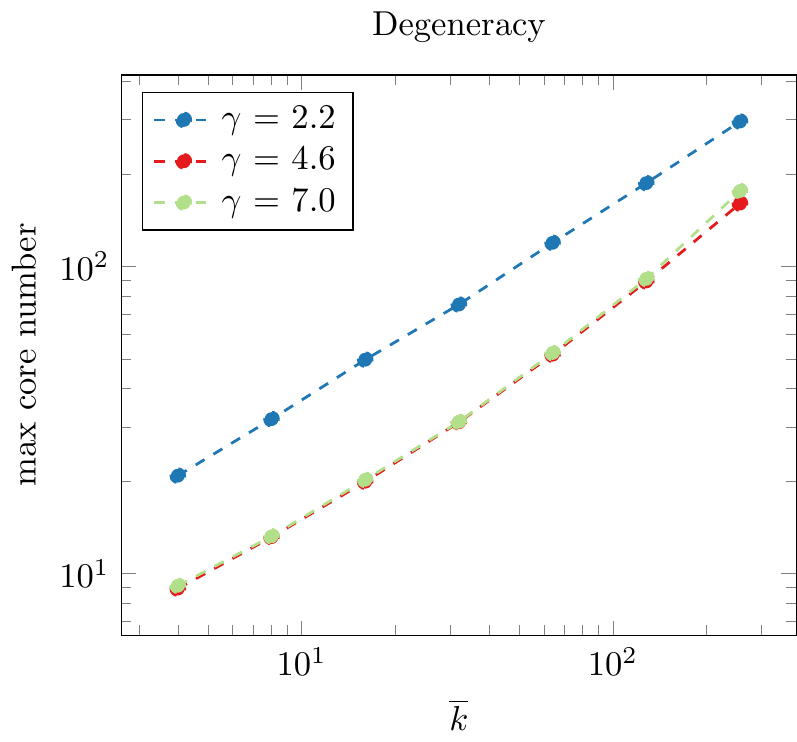}
        % \caption{Max core number}
        \label{plot:comparison-degen}
      \end{subfigure}

      \caption{Comparison of degree assortativity and degeneracy for the implementation of \cite{Aldecoa2015} (left) and our implementation (right).
      Degree assortativity describes whether vertices have neighbors of similar degree.
      A value near 1 signifies subgraphs with equal degree, a value of -1 star-like structures.
      $k$-Cores, in turn, are a generalization of connected components and result from iteratively peeling away vertices of degree $k$ and assigning to each vertex the core number of the innermost core it is contained in. Degeneracy refers to the largest core number.
      Values are averaged over 100 runs.}
      \label{plot:properties-comparison-I}
    \end{figure}
    
    \begin{figure}[h!]
      \begin{subfigure}[t]{.4\linewidth}
        \includegraphics[width=\linewidth]{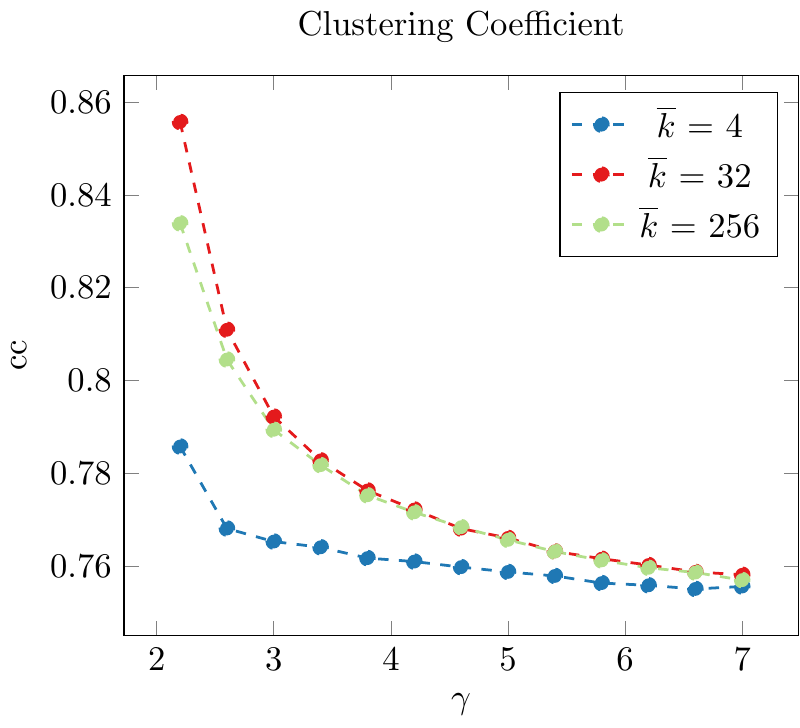}
        % \caption{Clustering coefficient}
        \label{plot:native-clustercoeff}
      \end{subfigure}
      \quad
      \begin{subfigure}[t]{.4\linewidth}
        \includegraphics[width=\linewidth]{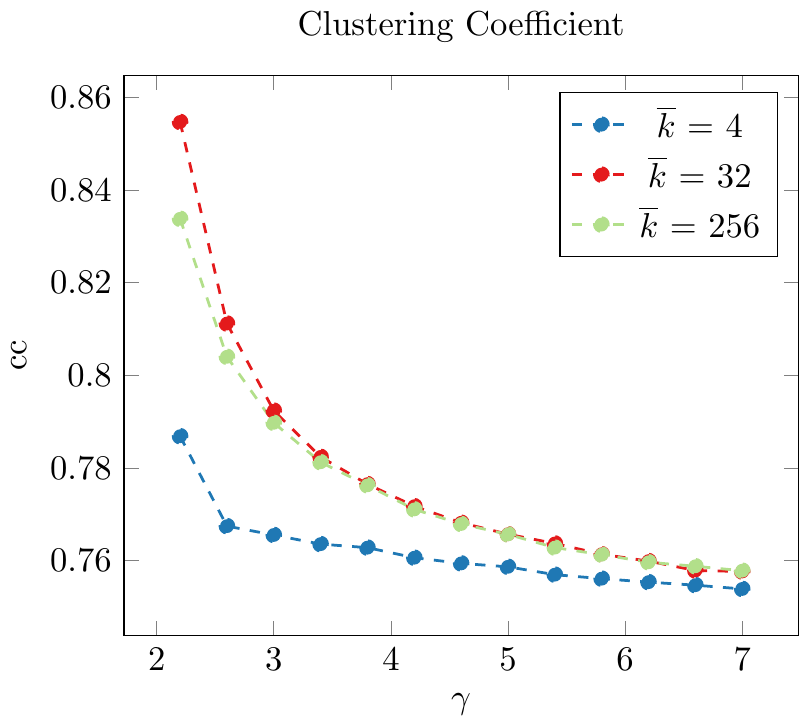}
        % \caption{Clustering coefficient}
        \label{plot:comparison-clustercoeff}
      \end{subfigure}

      \begin{subfigure}[t]{.4\linewidth}
        \includegraphics[width=\linewidth]{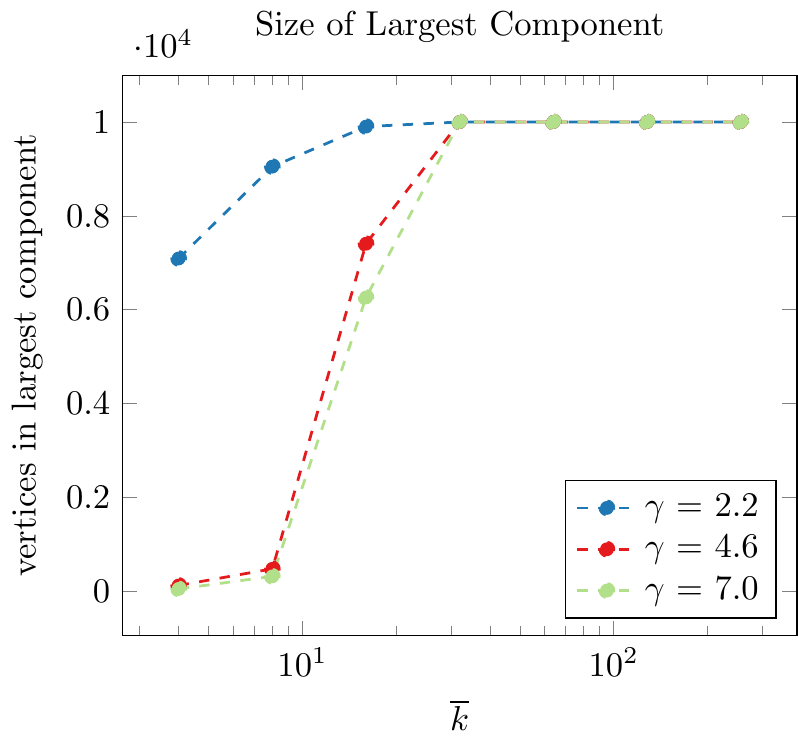}
        % \caption{Vertices in largest component}
        \label{plot:native-size-of-largest}
      \end{subfigure}
      \quad
      \begin{subfigure}[t]{.4\linewidth}
        \includegraphics[width=\linewidth]{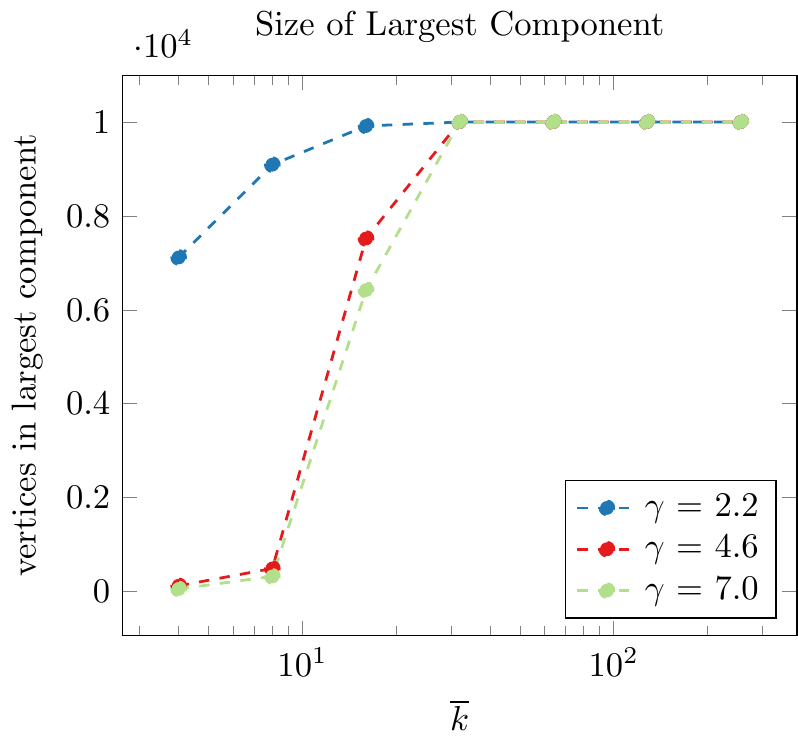}
        % \caption{Vertices in largest component}
        \label{plot:comparison-size-of-largest}
      \end{subfigure}

      \begin{subfigure}[t]{.4\linewidth}
        \includegraphics[width=\linewidth]{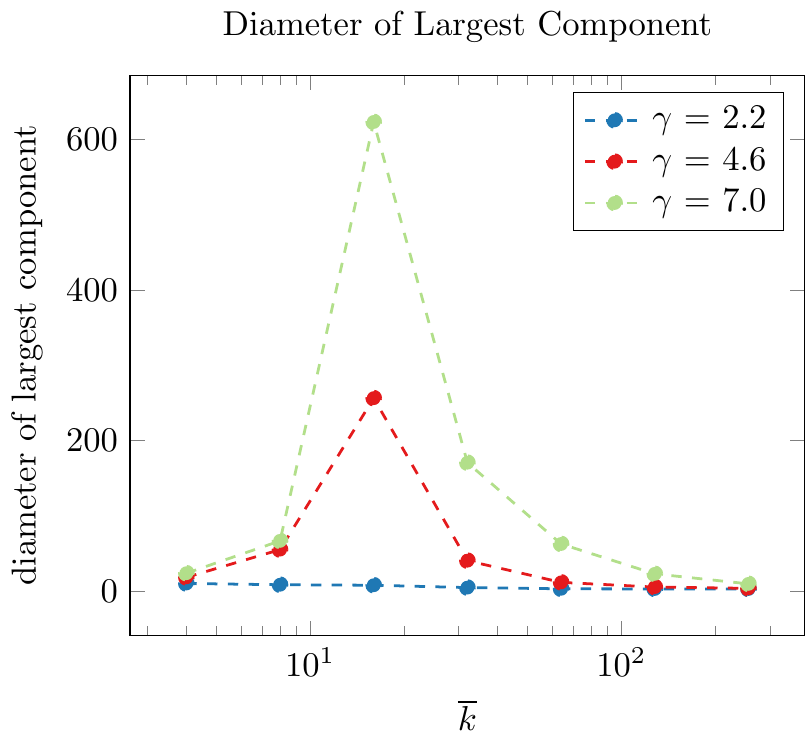}
        % \caption{Diameter of largest component}
        \label{plot:native-diameter}
      \end{subfigure}
      \quad
      \begin{subfigure}[t]{.4\linewidth}
        \includegraphics[width=\linewidth]{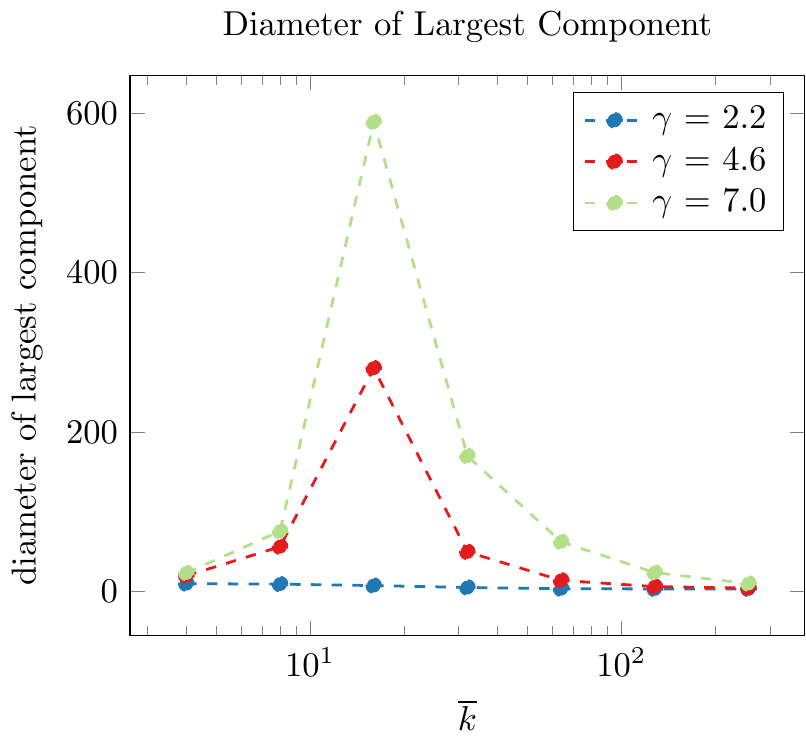}
        % \caption{Diameter of largest component}
        \label{plot:comparison-diameter}
      \end{subfigure}
      \caption{Comparison of clustering coefficients, size of largest component and diameter of largest components for the implementation of \cite{Aldecoa2015} (left) and our implementation (right).
      Values are averaged over 100 runs.}
      \label{plot:properties-comparison-II}
    \end{figure}

\clearpage
    \section{Consistency of Dynamic Model}
    \label{sec:dynamic-impl}

    \begin{figure}[h!]
      \begin{subfigure}[t]{.45\linewidth}
        \includegraphics[width=\linewidth]{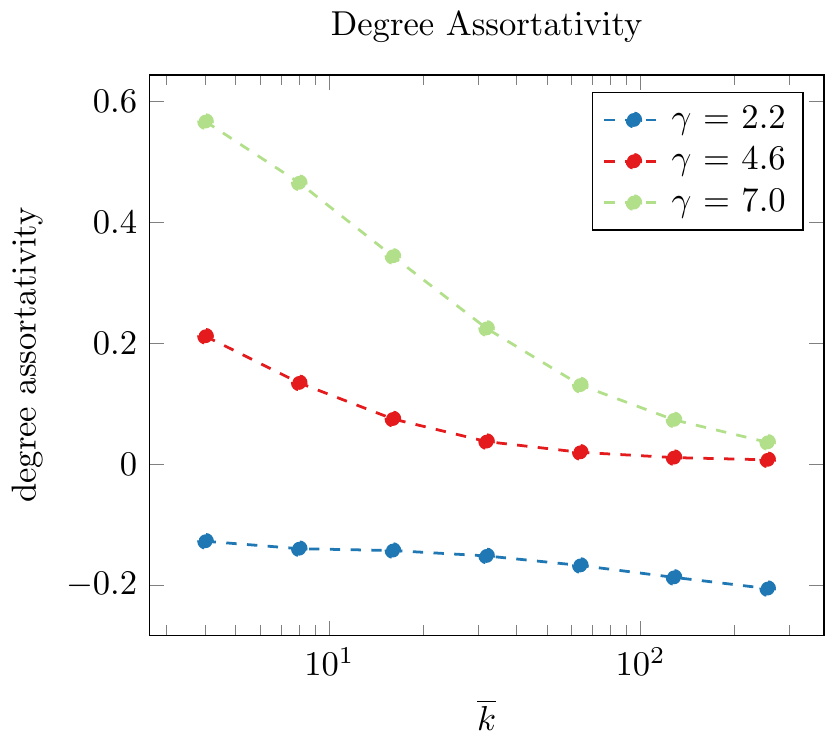}
        % \caption{Degree assortativity}
        \label{plot:dynamic-degass}
      \end{subfigure}
      \quad
      \begin{subfigure}[t]{.45\linewidth}
        \includegraphics[width=\linewidth]{plots/native-comparison/plot-degass}
        % \caption{Degree assortativity}
        \label{plot:static-degass}
      \end{subfigure}

      \begin{subfigure}[t]{.45\linewidth}
        \includegraphics[width=\linewidth]{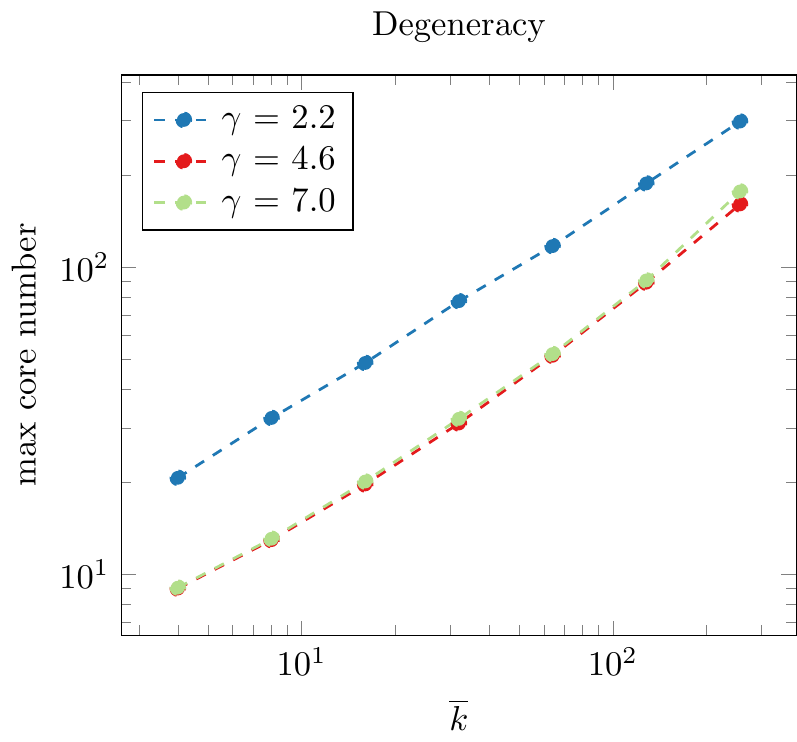}
        % \caption{Max core number}
        \label{plot:dynamic-degen}
      \end{subfigure}
      \quad
      \begin{subfigure}[t]{.45\linewidth}
        \includegraphics[width=\linewidth]{plots/native-comparison/plot-degen}
        % \caption{Max core number}
        \label{plot:static-degen}
      \end{subfigure}

      \caption{Comparison of degree assortativity and degeneracy for graphs with $10^4$ nodes, before and after one movement step.
      All nodes were moved, with $\tau_\phi \in (-1,1)$ and $\tau_r \in (-10,1)$ sampled randomly.
      Distribution of graphs after node movement are shown left, before node movement right.
      Values are averaged over 100 runs.}
      \label{plot:properties-dynamic-comparison-I}
    \end{figure}
    
    \begin{figure}[h!]
      \begin{subfigure}[t]{.4\linewidth}
        \includegraphics[width=\linewidth]{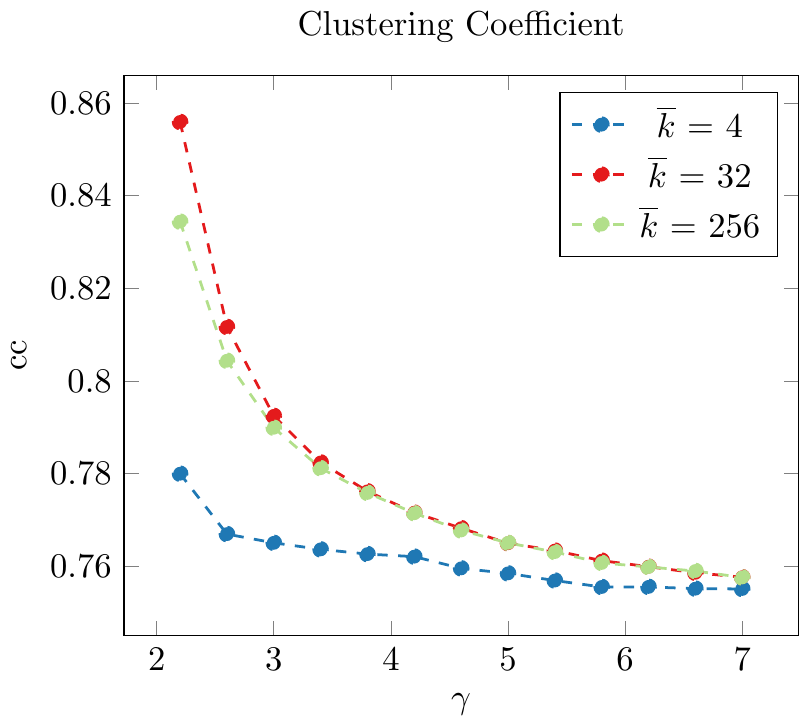}
        % \caption{Clustering coefficient}
        \label{plot:dynamic-clustercoeff}
      \end{subfigure}
      \quad
      \begin{subfigure}[t]{.4\linewidth}
        \includegraphics[width=\linewidth]{plots/native-comparison/plot-cc}
        % \caption{Clustering coefficient}
        \label{plot:static-clustercoeff}
      \end{subfigure}

      \begin{subfigure}[t]{.4\linewidth}
        \includegraphics[width=\linewidth]{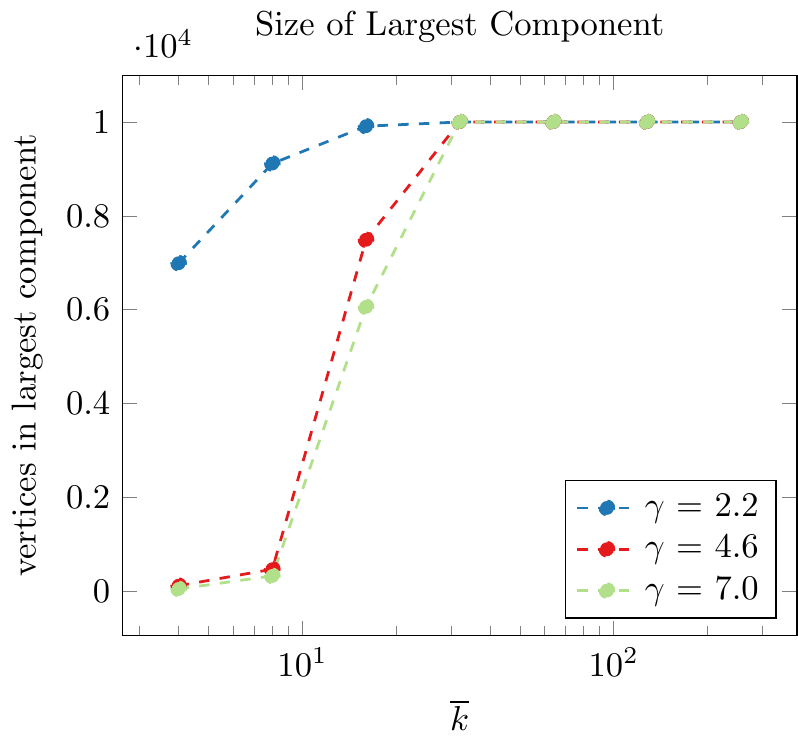}
        % \caption{Vertices in largest component}
        \label{plot:dynamic-size-of-largest}
      \end{subfigure}
      \quad
      \begin{subfigure}[t]{.4\linewidth}
        \includegraphics[width=\linewidth]{plots/native-comparison/plot-sizeOfLargest}
        % \caption{Vertices in largest component}
        \label{plot:static-size-of-largest}
      \end{subfigure}

      \begin{subfigure}[t]{.4\linewidth}
        \includegraphics[width=\linewidth]{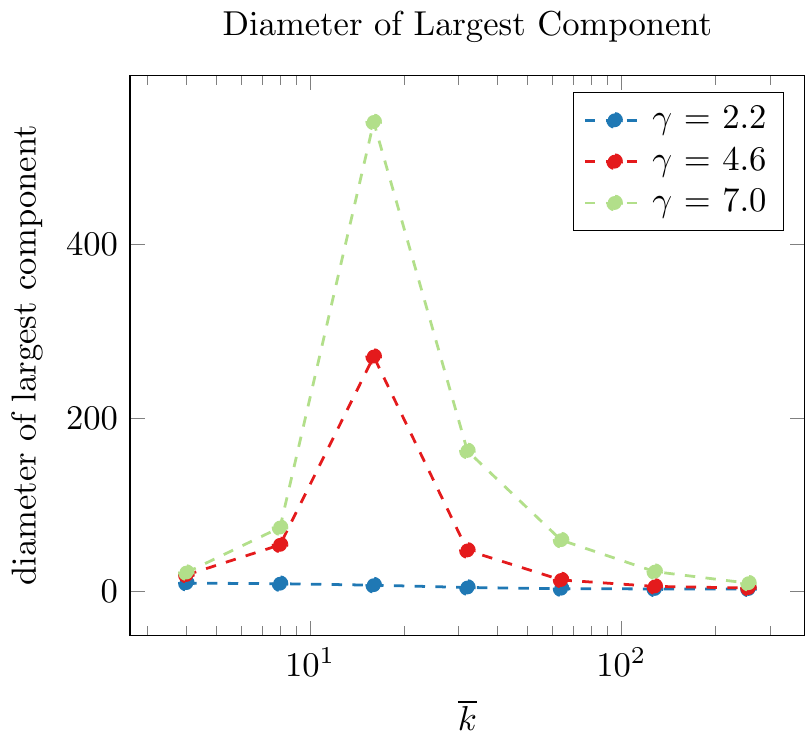}
        % \caption{Diameter of largest component}
        \label{plot:dynamic-diameter}
      \end{subfigure}
      \quad
      \begin{subfigure}[t]{.4\linewidth}
        \includegraphics[width=\linewidth]{plots/native-comparison/plot-diameter}
        % \caption{Diameter of largest component}
        \label{plot:static-diameter}
      \end{subfigure}
      \caption{Comparison of clustering coefficients, size of largest component and diameter of largest components for graphs with $10^4$ nodes, before and after one movement step.
      All nodes were moved, with $\tau_\phi \in (-1,1)$ and $\tau_r \in (-10,1)$ sampled randomly.
      Distribution of graphs after node movement are shown left, before node movement right.
      Values are averaged over 100 runs.}
      \label{plot:properties-dynamic-comparison-II}
    \end{figure}

  \end{document}